\documentclass[conference]{IEEEtran}
\IEEEoverridecommandlockouts
\usepackage{balance}
\usepackage[utf8]{inputenc} 
\usepackage[T1]{fontenc}
\usepackage{url}
\usepackage{ifthen}
\usepackage{cite,enumitem}
\usepackage[cmex10]{amsmath} % Use the [cmex10] option to ensure complicance
\usepackage{graphicx,mathtools}
\usepackage{amssymb,amsthm,verbatim,hyperref}
\usepackage{mathrsfs}
\usepackage{xcolor,tikz}
\usepackage[most]{tcolorbox} % For inline code styling
\usepackage{caption}
\tcbuselibrary{theorems}
\newcommand{\gaussbinom}[2]{\left[ \genfrac{}{}{0pt}{}{#1}{#2} \right]_2}

\definecolor{lightgray}{RGB}{224,224,224}

\newtheorem{theorem}{Theorem}
\newtheorem{exmp}{Example}
\newtheorem{corollary}{Corollary}%[theorem]
\newtheorem{definition}{Definition}
\newtheorem{remark}{Remark}
\newtheorem{proposition}{Proposition}

\newtheorem*{theorem*}{Theorem}

\newcounter{myboxctr}
\newtcolorbox{myrefbox}[2][]{%
  colback= white,
  colframe=gray,
  fonttitle=\bfseries,
  before upper={},
  title={\refstepcounter{myboxctr}\label{#1}Box~\themyboxctr:~#2}
}

\author{\IEEEauthorblockN{Hoang~Ly and Emina~Soljanin}
\IEEEauthorblockA{\textit{Department of Electrical and Computer Engineering, Rutgers University, USA} \\
\texttt{\{mh.ly,emina.soljanin}\}@rutgers.edu.}
\thanks{This work was supported in part by the NSF-BSF
grant FET-2120262, and has been presented at the Joint Mathematics Meetings (JMM), Washington, D.C., January 2026.}
}

\begin{document}
\title{Optimum 1-Step Majority-Logic Decoding of Binary Reed--Muller Codes
% {\footnotesize \textsuperscript{*}Note: Sub-titles are not captured for https://ieeexplore.ieee.org  and
% should not be used}
% \thanks{}
}

% \author{\IEEEauthorblockN{1\textsuperscript{st} Given Name Surname}
% \IEEEauthorblockA{\textit{dept. name of organization (of Aff.)} \\
% \textit{name of organization (of Aff.)}\\
% City, Country \\
% email address or ORCID}

%%% Several authors with up to three affiliations:
% \author{%
%   \IEEEauthorblockN{Hoang Ly and Emina Soljanin}
%   \IEEEauthorblockA{
%   Department of Electrical \& Computer Engineering, Rutgers University}
%   %, New Brunswick, NJ\\
%                     E-mail: \{\texttt{mh.ly;emina.soljanin}\}@rutgers.edu
%   % \and
%   % \IEEEauthorblockN{V. Lalitha}
%   % \IEEEauthorblockA{
% %Signal Processing and Communications Research Center\\ 
% % IIIT Hyderabad}
%                     % E-mail: lalitha.v@iiit.ac.in
% }

\maketitle

\begin{abstract}
The classical majority-logic decoder proposed by Reed for Reed--Muller codes \( \mathrm{RM}(r, m) \) of degree order \( r \) and length \( 2^m \), unfolds in \( r + 1 \) sequential steps, decoding message symbols from highest to lowest degree. Several follow-up decoding algorithms reduced the number of steps, but for a limited set of parameters, or at the expense of reduced performance, or relying on the existence of some combinatorial structures. We show that any \emph{one-step majority-logic decoder}—that is, a decoder performing all majority votes in one step simultaneously without sequential processing—can correct at most \( d_{\min}/4 \) errors for all values of \( r \) and \( m \), where \( d_{\min} \) denotes the code’s minimum distance. We then introduce a new hard-decision decoder that completes the decoding in a single step and attains this error-correction limit. It applies to all \( r \) and \( m \), and can be viewed as a parallel realization of Reed’s original algorithm, decoding all message symbols simultaneously. Remarkably, we also prove that the decoder is \emph{optimum} in the erasure setting: it recovers the message from any erasure pattern of up to \( d_{\min} - 1 \) symbols—the theoretical limit. To our knowledge, this is the first 1-step decoder for RM codes that achieves both optimal erasure correction and the maximum one-step error correction capability.
\end{abstract}
\section{Introduction}
Majority-logic decoding (MLD)~\cite{reed1954class,Massey1963Threshold} is a simple, practical, hard-decision decoding method. MLD recovers each message symbol by taking a majority vote over the values of multiple recovery sets (also known as repair groups)—collections of received symbols whose linear combination equals the target message symbol. Each such set provides an independent equation, and decoding succeeds if the majority remain uncorrupted.
% Known algorithms for decoding Reed--Muller codes \( \mathrm{RM}(r, m) \) can be loosely grouped into three categories. The first is majority-logic decoding (MLD), introduced in the landmark paper by Reed~\cite{reed1954class}. The second exploits the symmetry group of RM codes; see~\cite{second_order_decoding} and references therein, particularly the work of Sidel’nikov and Pershakov (in Russian). The third is recursive decoding, which leverages the \emph{Plotkin construction} \( (\mathbf{u},\, \mathbf{u} + \mathbf{v}) \) to decompose longer codes into shorter components, enabling bounded-distance decoding with low complexity; see~\cite{recursive_decoding:journal/tit/Dumer} and references therein. 
MLD was originally attractive because of its implementation simplicity, which requires only low-cost decoding circuitry, and extremely small decoding delay~\cite{Class_Of_MLDcodes}. %Today, despite the hardware and software advancements, there is a resurgence of interest in low-complexity decoders for latency-critical applications, which often favor hard-decision algorithms over soft-decision alternatives~\cite{RMdecoding:journals/tcom/BertramHH13}. Moreover, MLD frequently corrects more errors than its guaranteed radius, since it does not perform bounded-distance decoding~\cite{RMdecoding:journals/tit/Chen71}. Codes that can be decoded using such algorithms are known as \emph{majority-logic decodable codes}~\cite{Coding:books/PetersonW72}. Any majority-logic decodable code with minimum distance \( d_{\min} \) corrects up to \( d_{\min} - 1 \) erasures, which is the theoretical maximum for unique decoding~\cite{Coding:books/PetersonW72}. Historically, this optimal erasure tolerance made majority-logic decodable codes attractive for applications such as deep-space communication and memory protection~\cite{MLD_memory}. 
Today, despite major advances in hardware and software, there has been a resurgence of interest in low-complexity decoders for latency-critical applications, which often favor hard-decision algorithms over soft-decision alternatives~\cite{RMdecoding:journals/tcom/BertramHH13}. 
In fact, for certain families of codes—such as Grassmann codes—MLD remains the only known method for efficient practical decoding~\cite{Grassmannian_code_MJ}.
Moreover, MLD frequently corrects more errors than its guaranteed radius, since it does not perform bounded-distance decoding~\cite{RMdecoding:journals/tit/Chen71}. 
Codes that can be decoded using such algorithms are known as \emph{majority-logic decodable codes}~\cite{Coding:books/PetersonW72}. 
Any majority-logic decodable code with minimum distance \( d_{\min} \) can correct up to \( d_{\min} - 1 \) erasures, which is the theoretical maximum for unique decoding~\cite{Coding:books/PetersonW72}. 
Historically, this optimal erasure tolerance made majority-logic decodable codes attractive for applications such as deep-space communication and memory protection~\cite{MLD_memory,OS_MLD_memory_protection,MLD_LDPC}. More recently, MLD has also been linked to combinatorial designs and to the notion of the \emph{Service Rate Region}, a metric that quantifies the utility of linear codes for optimizing data access~\cite{SRR_Design:lySV2025}. Furthermore, MLD has re-emerged as the efficient decoding method for \emph{locally recoverable codes}—a class of codes that has attracted significant attention due to their applications in distributed storage and private information retrieval~\cite{MLD_LRC_26}.

One of the earliest known majority-logic decodable codes is the class of Reed--Muller (RM) codes, for which the MLD algorithm was introduced as early as 1954 in the landmark paper by Reed~\cite{reed1954class}. This is a sequential algorithm that decodes the codeword symbols of an \( \mathrm{RM}(r, m) \) code—of order \( r \) and length \( 2^m \)—in exactly \( r + 1 \) steps. Several subsequent papers aimed to reduce the number of sequential steps in Reed’s decoding. \textcolor{black}{Chen made a significant advancement in 1971~\cite{RMdecoding:journals/tit/Chen71} by showing that finite geometry codes are majority-logic decodable in at most three steps; in particular, Reed--Muller codes $\mathrm{RM}(r,m)$ can be decoded in two steps by observing that certain stages in Reed’s original decoding algorithm can be skipped. Chen’s approach associates each codeword symbol with a $0$-dimensional flat (a point) in the corresponding Euclidean geometry. He then constructs carefully chosen families of $(r+1)$-flats whose pairwise intersections are $r$-flats, and whose $r$-flats in turn intersect at the underlying $0$-flats corresponding to codeword symbols. This hierarchical intersection structure enables majority-logic decisions to be performed in two stages. The resulting decoding algorithm can correct up to $\lfloor (d_{\min}-1)/2 \rfloor$ errors, i.e., approximately half the minimum distance, which is the theoretical limit for unique decoding. Nevertheless, the algorithm remains inherently sequential (requiring two decoding steps) and is restricted to the regime $r \le m/2$. Hardware-efficient refinements of Chen’s method were later proposed in~\cite{RMdecoding:journals/tcom/BertramHH13}, but these implementations remain two-step and subject to the same parameter constraints.}

Alternatively, Peterson and Weldon~\cite{Coding:books/PetersonW72} proposed a two-step MLD algorithm for geometric codes that also achieves half-distance error correction. More recently, the method was adapted in~\cite{MLDecoding_subspace:journals/tit/CruzW21} to leverage subspace designs for reducing the decoding complexity, though the decoding process still proceeds in two steps. \textcolor{black}{In parallel, MLD has been applied to efficiently decode codes with weaker algebraic or geometric structure, such as low-density parity-check (LDPC) codes~\cite{Analysis_MLD_LDPC_faulty,MLD_LDPC}. 
More recently, an MLD algorithm was developed for polar codes—a class of capacity-achieving codes whose existing efficient decoders are primarily soft-decision based. 
This method exploits the structural relation between polar and Reed--Muller codes, supports both soft- and hard-decision decoding, and achieves error-correction performance up to the theoretical bound~\cite{Polar_code_MLD}.}
%Since errors in received codeword symbols affect all recovery sets that contain them, the effectiveness of MLD depends critically on the overlap structure among these sets.

A different approach connects MLD capability to the combinatorial structure of the code. Rudolph was the first to connect finite geometries to the construction
of majority-logic decodable codes. He showed that if the set of supports of codewords in the dual code forms a \( (v, b, r, k, \lambda) \) balanced incomplete block design (BIBD), then up to \( \left\lfloor \frac{r}{2\lambda} \right\rfloor \)
errors can be corrected using MLD in only one step
\cite{RMdecoding:journals/tit/Rudolph67}. This result was later slightly modified and improved in~\cite{Rudolph_modified}. The key idea is that when the dual code supports a BIBD, (possibly overlapping) recovery sets for all message symbols can be constructed, enabling majority voting to be performed in parallel across all symbols. We observe that allowing overlaps increases the number of available recovery sets, but an error in a received coordinate corrupts every recovery set that includes it. Therefore, the effectiveness of MLD depends critically on the intersection structure among the recovery sets. Rahman and Blake generalized this approach from BIBDs to arbitrary combinatorial $t$-designs ($t>2$) and applied it to analyze the decoding capability of one-step majority-logic decoding (1S-MLD) for first-, second-, and third-order RM codes~\cite{RMdecoding:journals/tit/RahmanI75}. Their analysis relies on the fact that the supports of codewords of a fixed weight in Reed–Muller codes form a $3$-design, which follows from the classical result that Reed–Muller codes are invariant under a triply transitive permutation group~\cite{KasamiLin1966AFCRL}.

These results established a connection between the dual code structure and the performance of 1S-MLD. \textcolor{black}{Their key limitation is that no $t$-designs with strength $t > 5$ are currently known. For RM codes in particular, only $3$-designs supported by codewords of the dual are known~\cite{Design_codes:books/Designs_and_Codes}, which implies that these one-step decoding algorithms are constrained to work only with codes of limited minimum distance—at most $9$ in the best-known cases~\cite{RMdecoding:journals/tit/RahmanI75}.} In their work, Rahman and Blake conjectured that the advantage of their method over the modified Rudolph algorithm would become more pronounced if the $t$ parameter of the $t$-designs increases, i.e., for $t=4$ or $5$. Thus, the performance of 1S-MLD is fundamentally constrained by its intrinsic reliance on the combinatorial structure of the underlying dual code. For many codes, 1S-MLD can correct far fewer errors than one-half of the minimum distance~\cite[Ch.~10]{Coding:books/PetersonW72}. In fact, a 1S-MLD that works for arbitrary code parameters and corrects a number of errors that grows with the minimum distance of the code (beyond a fixed constant) remains unknown, despite being highly coveted. For a concise historical overview, see~\cite[Ch.~5]{Huffman2021Concise}.

Overlapping recovery sets of code symbols play a role in an entirely different line of work related to data access from coded storage \cite{SRR:journals/tit/AktasJKKS21}.
In particular, we characterized the recovery sets of RM codes in \cite{SRR_RM_ISIT,SRR:preprint/arxiv/LySL25}, and further studied their combinatorial relations and connections to combinatorial designs in \cite{SRR_Design:lySV2025}. Here, we build on these results to construct universally optimal 1S-MLD for RM codes.

% \textbf{Our contributions.} Leveraging our previous landmark results that characterize the recovery sets of message symbols in RM codes and their corresponding geometric structure, we construct the first 1-step majority-logic decoder for RM codes that applies to all parameters \( r \) and \( m \). We prove that this decoder corrects up to \( d_{\min}/4 \) errors, and establish that this is the maximum achievable by any 1-step majority-logic decoder (1S-MLD) for RM codes—specializing a known general bound for linear codes. Our contribution lies in demonstrating that this bound is tight for RM codes. We further establish that the decoder is also optimal for erasure correction, successfully recovering from any erasure pattern of weight up to \( d_{\min} - 1 \), which is the maximum possible. As a crucial part of this result, we present and prove a new result on the minimum size of a transversal for the family of all complements of a fixed subspace in Euclidean geometry.
\textcolor{black}{\paragraph*{\textbf{Our contributions.}}
Leveraging our previous results characterizing the recovery sets of message symbols in Reed--Muller (RM) codes and their underlying geometric structure, we make the following contributions:
\begin{itemize}
    \item \textbf{Error decoding.}
    We construct the first one-step majority-logic decoder for RM codes that applies uniformly to all parameters \( r \) and \( m \).
    We prove that this decoder can correct up to \( d_{\min}/4 \) errors (Theorem~\ref{thm:optimal_error_decoder}).
    Moreover, by specializing a known general bound for linear codes to the RM setting, we establish that \( d_{\min}/4 \) is the maximum number of errors that any 1S-MLD can correct for RM codes, and show that this bound is tight (Proposition~\ref{prop:MLD_limit}).
    \item \textbf{Erasure decoding.}
    We further prove that the proposed decoder is optimal in the erasure setting: it can recover from any erasure pattern of weight up to \( d_{\min}-1 \), which is the theoretical maximum (Theorem~\ref{thm:optimal_erasure_decoder}).
    A key component of this result is a new combinatorial-geometric theorem that characterizes the minimum size of a transversal for the family of all complements of a fixed subspace in Euclidean geometry (Theorem~\ref{thm:optimal_transversal}).
\end{itemize}
}
Table~\ref{tab:MLD_compare} compares our proposed method with other MLD algorithms for RM codes. Our findings also highlight a deeper connection between error and erasure correction capabilities in RM codes, relating to and complementing recent insights from~\cite{RM_random_erasure_error,RM_efficient_decoding}.

% We show that any such decoder is fundamentally limited to correcting at most \( d_{\min}/4 \) errors, where $d_{\min}$ denotes the minimum Hamming distance of the code. We then construct a decoder that achieves this bound and demonstrate its optimality in both error and erasure settings, including the recovery of up to \( d_{\min} - 1 \) erasures—the theoretical maximum. Our findings also highlight a deeper connection between error and erasure correction capabilities in RM codes, complementing recent insights from~\cite{RM_random_erasure_error,RM_efficient_decoding}.

This paper is organized as follows. Section~\ref{sec:nomenclature} introduces definitions and notations. Section~\ref{sec:preliminary} provides background on RM codes and MLD. Section~\ref{sec:recovery} reviews the structure of recovery sets for message symbols, based on recent works in~\cite{SRR_RM_ISIT,SRR:preprint/arxiv/LySL25}. Section~\ref{sec:RM_MLD} presents our 1S-MLD construction and proves its optimality for both error and erasure decoding. Section~\ref{sec:conclusion} concludes the paper.

% \hnote{Table~\ref{tab:MLD_compare} is newly added.}
\begin{table*}[t]
\caption{Comparisons of different MLD algorithms for $\mathrm{RM}(r, \, m)$ codes}
    \centering
    \begin{tabular}{ccc}
         Methods & Correcting performance & Characteristics \\
         \hline
         Reed's decoding~\cite{reed1954class}& $\lfloor\frac{d_{\text{min}}}{2}\rfloor$ errors; $d_{\text{min}}$-1 erasures & Works in $(r+1)$ steps, with all parameters $(r,\, m)$.\\
         \hline
         Chen's algorithm~\cite{RMdecoding:journals/tit/Chen71}& $\lfloor\frac{d_{\text{min}}}{2}\rfloor$ errors; $d_{\text{min}}$-1 erasures & Sequential in 2 steps, limited to $r \le m/2$ regimes. \\
         \hline
         Bertram et al.~\cite{RMdecoding:journals/tcom/BertramHH13} & $\lfloor\frac{d_{\text{min}}}{2}\rfloor$ errors; $d_{\text{min}}$-1 erasures & Sequential in 2 steps, limited to $r \le m/2$ regimes.  \\
         \hline
         Rahman and Blake~\cite{RMdecoding:journals/tit/RahmanI75} & No known general formula; & Works in 1 step, limited to 1st-, 2nd-, and 3rd-order RM codes.\\
        & Below  $\lfloor\frac{d_{\text{min}}}{2}\rfloor$ & Limited to the existence of $t$-designs (no designs with $t > 5$ are known.)\\
        \hline
        Our method & $\lfloor\frac{d_{\text{min}}}{4}\rfloor$ errors; $d_{\text{min}}$-1 erasures & Work in 1 step, with all parameters $(r,\, m)$.\\
        
    \end{tabular}
    \label{tab:MLD_compare}
\end{table*}

\section{Basic Definitions and Notations}
\label{sec:nomenclature}
Our notations are standard in coding theory, algebra, and finite geometry. \(\mathrm{RM}(r, m)\) stands for
the binary Reed--Muller code of order \(r\) and length \(n=2^m\). \(\mathbb{F}_q\) is the finite field over a prime or prime power \(q\). A linear \(q\)-ary code \(\mathcal{C}\) with parameters \([n, k, d]_q\), is a \(k\)-dimensional subspace of the \(n\)-dimensional vector space \(\mathbb{F}_q^n\) with the minimum distance $d$. We denote as $c_j$ the $j$-th column of the code's generator matrix, $1 \le j \le n$. Hamming weight of a codeword \(\boldsymbol{x}\) in \(\mathcal{C}\) is denoted as \(\text{wt}(\boldsymbol{x})\). %The symbols \(\boldsymbol{0}_k\) and \(\boldsymbol{1}_k\) denote the all-zero and all-one column vectors of length \(k\), respectively. 
The standard basis (column) vector with a one at position \(i\) and a zero elsewhere is \(\mathbf{e}_i\). The transpose of a vector $\boldsymbol{v}$ is $\boldsymbol{v}^{\top}$. \(\mathrm{Supp}(\boldsymbol{x})\) denotes the support of codeword \(\boldsymbol{x}\). The set of positive integers not exceeding \(i\) is denoted as \([i]\). %Similarly, \([a, b]\) represents the set of integers between \(a\) and \(b\), where \(a, b \in \mathbb{N}\) and \(a < b\). 

In the \emph{Euclidean geometry} $\mathrm{EG}(m, 2)$ of dimension $m$ over $\mathbb{F}_2$, an \emph{$r$-flat} is defined as an affine subspace of dimension $r$, consisting of $2^r$ points. Formally, it is a coset of an 
$r$-dimensional linear subspace of $\mathbb{F}_2^m$, meaning it can be expressed as $v + V$, where $V$ is an $r$-dimensional subspace and $v \in \mathbb{F}_2^m$ is a shift vector. Flats generalize points, lines, and planes in the finite field setting. The Gaussian binomial coefficient $\gaussbinom{m}{r} = \prod\limits_{i=0}^{r-1}\dfrac{1-2^{m-i}}{1-2^{i+1}}$ counts the number of \( r \)-dimensional subspaces of \( \mathbb{F}_2^m \), i.e., the cardinality of the \emph{Grassmannian} \( \mathcal{G}_2(m, r) \). By convention, \( \gaussbinom{m}{0} := 1 \). The notation \( F \le G \) indicates that \( F \) is a subspace of \( G \), meaning \( F \subseteq G \) and both are linear subspaces of \( \mathbb{F}_2^m \). 

% A \( t\text{--}(n,k,\lambda) \) block design (or, a \( t \)-design) is a combinatorial structure consisting of a set \( P \) of \( n \) elements (called \emph{points}) and a collection of \( k \)-element subsets of \( P \) (called \emph{blocks}), such that every \( t \)-subset of \( P \) is contained in exactly \( \lambda \) blocks. 

In the projective space \(\mathrm{PG}(n, q)\), a \(k\)-space is a \(k\)-dimensional projective subspace, generalizing the notions of points (\(0\)-spaces), lines (\(1\)-spaces), and planes (\(2\)-spaces). A set \( B \) of points in \( \mathrm{PG}(n, q) \) is called a transversal (or a \emph{blocking set}) with respect to \( k \)-spaces if every \( k \)-space of \( \mathrm{PG}(n, q) \) contains at least one point from \( B \), see, e.g., 
\cite{Blocking_subspaces}. A transversal $B$ is minimal if no proper subset of $B$ is itself a transversal.

The abbreviation 1S-MLD is used to refer to either the one-step majority-logic \emph{decoding procedure} or the associated \emph{decoder}. The intended meaning will be clear from the context.

\section{Reed--Muller Codes Preliminaries}\label{sec:preliminary}
We begin by introducing Reed–Muller (RM) codes, exploring their connection to Euclidean geometry, and reviewing the classical Reed decoding algorithm. This is followed by a brief overview of 1S-MLD. Finally, we present a fundamental limit on the number of errors that 1S-MLD can correct in RM codes. Throughout, we adopt standard notations and definitions from~\cite[Ch. 13]{Coding:books/MacWilliamsS77} and~\cite[Ch.~10]{Coding:books/PetersonW72}. This section lays the groundwork for establishing major results that follow.
% establishing a connection with dual codewords, and examine how these concepts aid in analyzing the Service Rate Region (SRR).
\subsection{Reed--Muller Codes}
Let \( v_1, \dots, v_m \) be \( m \) binary variables, and \( \boldsymbol{v} = (v_1, \dots, v_m) \) binary \( m \)-tuples. Consider a Boolean function \( f(\boldsymbol{v}) = f(v_1, \dots, v_m) \). We denote by \( \boldsymbol{f} \) the vector of all evaluations of \( f \) on its $2^m$ possible arguments \( \boldsymbol{v} \).

\begin{definition}[see, e.g.,{\cite[Ch.~13]{Coding:books/MacWilliamsS77}}]
The \( r \)-th order binary Reed--Muller code \( \mathrm{RM}(r, m) \) of length \( n = 2^m \), for \( 0 \leq r \leq m \), consists of all vectors \( \boldsymbol{f} \) where \( f(\boldsymbol{v}) \) is a Boolean function that can be expressed as a polynomial of degree at most \( r \).
\label{def:RM_code}
\end{definition}

The Reed--Muller code is a linear code with length \( n = 2^m \), dimension \( k = \sum_{i=0}^r \binom{m}{i} \triangleq \binom{m}{\le r}\), and minimum distance \( d = 2^{m-r} \). Therefore, \( \mathrm{RM}(r, m) \) is characterized by the parameters \( [n, k, d] = [2^m, \binom{m}{\le r}, 2^{m-r}] \). When $m \ge r+1$, the dual code of $\mathrm{RM}(r, \,m)$ is $\mathrm{RM}(m-r-1,\, m)$~\cite{Coding:books/MacWilliamsS77}. In the sequel, we assume $m \ge r+1$, ensuring that the dual $\mathrm{RM}(m-r-1,\, m)$ is always well defined. Note that when $r = m$, the code is a trivial $[n, n, 1]_2$ code with zero redundancy, i.e., the entire ambient space. 

% \textcolor{blue}{\textbf{Remark 1.} In the sequel, we assume that $m\ge r+1$.}

% To illustrate, consider the first-order RM code \( \mathrm{RM}(1, 3) \) of length \( 2^3 = 8 \), which contains 16 codewords of the form:
% \[
% a_0 \boldsymbol{1} + a_1 \boldsymbol{v}_1 + a_2 \boldsymbol{v}_2 + a_3 \boldsymbol{v}_3, \quad \text{where } a_i \in \mathbb{F}_2,
% \]
% where
% \[
% \begin{bmatrix}
% \boldsymbol{1} \\
% \boldsymbol{v}_3 \\
% \boldsymbol{v}_2 \\
% \boldsymbol{v}_1 \\
% \end{bmatrix}
% =
% \begin{bmatrix}
% 1 & 1 & 1 & 1 & 1 & 1 & 1 & 1 \\
% 0 & 0 & 0 & 0 & 1 & 1 & 1 & 1 \\
% 0 & 0 & 1 & 1 & 0 & 0 & 1 & 1 \\
% 0 & 1 & 0 & 1 & 0 & 1 & 0 & 1 \\
% \end{bmatrix}.
% \]
% This matrix serves as the \textit{generator matrix} of \( \mathrm{RM}(1, 3) \).

For example, the generator matrix for \( \mathrm{RM}(2, 4) \) is given by
\begin{equation}
\boldsymbol{G}_{\mathrm{RM}(2,4)} = 
\left[
\begin{smallmatrix}
\vrule width 0pt depth 1pt\\
\boldsymbol{1}\\
\boldsymbol{v}_4\\
\boldsymbol{v}_3\\
\boldsymbol{v}_2\\
\boldsymbol{v}_1\\
\boldsymbol{v}_3\boldsymbol{v}_4\\
\boldsymbol{v}_2\boldsymbol{v}_4\\
\boldsymbol{v}_1\boldsymbol{v}_4\\
\boldsymbol{v}_2\boldsymbol{v}_3\\
\boldsymbol{v}_1\boldsymbol{v}_3\\
\boldsymbol{v}_1\boldsymbol{v}_2
\vrule width 0pt depth 4pt
\end{smallmatrix}
\right]
= 
\left[
\begin{smallmatrix}
\vrule width 0pt depth 2pt\\
1 & 1 & 1 & 1 & 1 & 1 & 1 & 1 & 1 & 1 & 1 & 1 & 1 & 1 & 1 & 1\\
0 & 0 & 0 & 0 & 0 & 0 & 0 & 0 & 1 & 1 & 1 & 1 & 1 & 1 & 1 & 1\\
0 & 0 & 0 & 0 & 1 & 1 & 1 & 1 & 0 & 0 & 0 & 0 & 1 & 1 & 1 & 1\\
0 & 0 & 1 & 1 & 0 & 0 & 1 & 1 & 0 & 0 & 1 & 1 & 0 & 0 & 1 & 1\\
0 & 1 & 0 & 1 & 0 & 1 & 0 & 1 & 0 & 1 & 0 & 1 & 0 & 1 & 0 & 1\\
0 & 0 & 0 & 0 & 0 & 0 & 0 & 0 & 0 & 0 & 0 & 0 & 1 & 1 & 1 & 1\\
0 & 0 & 0 & 0 & 0 & 0 & 0 & 0 & 0 & 0 & 1 & 1 & 0 & 0 & 1 & 1\\
0 & 0 & 0 & 0 & 0 & 0 & 0 & 0 & 0 & 1 & 0 & 1 & 0 & 1 & 0 & 1\\
0 & 0 & 0 & 0 & 0 & 0 & 1 & 1 & 0 & 0 & 0 & 0 & 0 & 0 & 1 & 1\\
0 & 0 & 0 & 0 & 0 & 1 & 0 & 1 & 0 & 0 & 0 & 0 & 0 & 1 & 0 & 1\\
0 & 0 & 0 & 1 & 0 & 0 & 0 & 1 & 0 & 0 & 0 & 1 & 0 & 0 & 0 & 1
\vrule width 0pt depth 3pt
\end{smallmatrix}
\right],\label{eq:generator_matrix}
\end{equation}
in which $\boldsymbol{v}_i\boldsymbol{v}_j$ denotes the element-wise product of the row vectors $\boldsymbol{v}_i$ and $\boldsymbol{v}_j$. 

\subsection{A Geometric View}
Reed–Muller (RM) codes can be elegantly described using finite geometry, specifically \( \mathrm{EG}(m, 2) \), the Euclidean geometry of dimension \( m \) over \( \mathbb{F}_2 \). This space consists of \( 2^m \) points \( P_i \) for \( i = 1, \dots, 2^m \), each identified with a length-$m$ binary coordinate vector \( \boldsymbol{t} = (t_1, \dots, t_m) \in \mathbb{F}_2^m \). For any subset of points \( S \subseteq \mathrm{EG}(m, 2) \), its \textit{incidence vector} \( {\boldsymbol{\chi}}(S) \) is defined as
\[
\boldsymbol{\chi}(S)_j = 1 \text{ if } P_j \in S, \text{ and } \boldsymbol{\chi}(S)_j = 0 \text{ otherwise},
\]
where \( P_j \) denotes the \( j \)-th point. %Similarly, a subset \( S \) of columns of $\boldsymbol{G}$ is a recovery set for \( o_i \) if it satisfies \( \boldsymbol{G} \cdot \boldsymbol{\chi}(S)^{\top} = \mathbf{e}_i \), where \( \boldsymbol{G} \) is the generator matrix, and \( \mathbf{e}_i \) the standard basis vector for \( o_i \).
Codewords of \( \mathrm{RM}(r, m) \) correspond to incidence vectors of specific subsets in \( \mathrm{EG}(m, 2) \). For example, in \( \mathrm{EG}(4, 2) \), subset \( S = \{P_5, P_6, P_7, P_8, P_{13}, P_{14}, P_{15}, P_{16}\} \) has incidence vector \( \boldsymbol{\chi}(S) = 0000111100001111 \), a codeword of \( \mathrm{RM}(2, 4) \). The numbering of points \( P_1,\dots, P_{16} \) is assumed to follow the order in Table~\ref{table:points}, where \( P_1 \) represents the origin (point with all coordinates equal zero).

For each binary vector \( \boldsymbol{x} = [x_1, \dots, x_{2^m}] \in \mathbb{F}_2^{2^m} \), the coordinate \( x_j \) corresponds to the point \( P_j \). That is, \( \boldsymbol{x} \) serves as the incidence vector of a subset \( S(\boldsymbol{x}) \subseteq \mathrm{EG}(m, 2) \) (i.e., $\boldsymbol{x} = \boldsymbol{\chi}(S(\boldsymbol{x}))$), where \( x_j = 1 \) indicates that \( P_j \in  S(\boldsymbol{x}) \). The size of this subset equals the Hamming weight \( \text{wt}(\boldsymbol{x}) \). This notation enables the interpretation of Reed--Muller codewords as incidence vectors over the point set of the Euclidean geometry, highlighting the connection between RM codes and geometry, as seen in the following theorem.
\begin{theorem}[see, e.g.,{~\cite[Ch.~13, Thm.~5, 8, 12]{Coding:books/MacWilliamsS77}}]
% (\hspace{-0.3mm}\cite{Coding:books/MacWilliamsS77}, Ch.\ 13, Theorem 8 \& 12) 
Let $f$ be a minimum weight codeword of \( \mathrm{RM}(r, m) \), say $f = \boldsymbol{\chi}(S)$. Then $S$ is an $(m-r)$-dimensional flat in $\mathrm{EG}(m, 2)$. Equivalently, the minimum-weight codewords of \( \mathrm{RM}(r, m) \) are precisely the incidence vectors of the \((m - r)\)-flats in \( \mathrm{EG}(m, 2) \).

The set of incidence vectors of all \((m - r)\)-flats in \( \mathrm{EG}(m, 2) \) generates the Reed--Muller code \( \mathrm{RM}(r, m) \).
% The codewords of minimum weight in \( \mathrm{RM}(r, m) \) are precisely the incidence vectors of the \((m - r)\)-flats in \( \mathrm{EG}(m, 2) \), and these flats generates the Reed--Muller code \( \mathrm{RM}(r, m) \).
\label{thm:MinWeightRM}
\end{theorem}
\begin{table*}[t]
% \vspace{-13mm}
\caption{$2^4 = 16$ Points in $\mathrm{EG}(4, \,2)$ with their coordinate vectors.} % Add your caption here
\label{table:points} % Keep the label here to allow referencing
\begin{center}
\vspace{-4mm}
\begin{tabular}{c c c c c c c c c c c c c c c c c}
&   $P_1$ & $P_2$ & $P_3$ & $P_4$ & $P_5$ & $P_6$ & $P_7$& $P_8$ & $P_9$ & $P_{10}$ & $P_{11}$ & $P_{12}$ & $P_{13}$ & $P_{14}$ & $P_{15}$ & $P_{16}$\\
 $\boldsymbol{v}_4,\, t_1$ &  0 & 0 & 0 & 0 & 0 & 0 & 0 & 0 & 1 & 1 & 1 & 1 & 1 & 1 & 1 & 1\\
 $\boldsymbol{v}_3, \,t_2$ & 0 & 0 & 0 & 0 & 1 & 1 & 1 & 1 & 0 & 0 & 0 & 0 & 1 & 1 & 1 & 1 \\
 $\boldsymbol{v}_2, \,t_3$ & 0 & 0 & 1 & 1 & 0 & 0 & 1 & 1 & 0 & 0 & 1 & 1 & 0 & 0 & 1 & 1 \\
 $\boldsymbol{v}_1, \,t_4$ & 0 & 1 & 0 & 1 & 0 & 1 & 0 & 1 & 0 & 1 & 0 & 1 & 0 & 1 & 0 & 1 \\
\end{tabular}
\end{center}
\vspace{-5mm}
\end{table*}
\subsection{Reed's Decoding Algorithm}
A well-known  MLD algorithm for Reed–Muller (RM) codes is the \emph{Reed decoding algorithm} proposed by I. Reed~\cite{reed1954class}. We illustrate its operation using the second-order RM code \( \mathrm{RM}(2, 4) \), which has parameters \([16, 11, 4]\).

The generator matrix of \( \mathrm{RM}(2, 4) \) encodes the message
% \[
% \boldsymbol{a} := \left[a_0,\, a_4,\, a_3,\, a_2,\, a_1,\, a_{34},\, a_{24},\, a_{14},\, a_{23},\, a_{13},\, a_{12}\right]
% \]
\[
    \boldsymbol{a} = [\underbrace{a_0}_{0^{\text{th}}\text{-}}, \underbrace{a_4, \ a_3, \ a_2, \ a_1,}_{1^{\text{st}}\text{-}} \underbrace{a_{34}, \ a_{24}, \ a_{14}, \ a_{23}, \ a_{13}, \ a_{12}}_{2^{\text{nd}} \text{-order symbols}}]
\]
into the codeword
\begin{align}
\boldsymbol{x} &= \boldsymbol{a} \cdot \boldsymbol{G}
= a_0 \boldsymbol{1} + a_4 \boldsymbol{v}_4 + \dots + a_1 \boldsymbol{v}_1 \notag \\
&\quad + a_{34} \boldsymbol{v}_3 \boldsymbol{v}_4 + \dots + a_{12} \boldsymbol{v}_1 \boldsymbol{v}_2
:= [x_1,\, x_2,\, \dots,\, x_{16}],
\label{eq:message_to_codeword}
\end{align}
where \( \boldsymbol{v}_i \in \mathbb{F}_2^{16} \) is defined in~\eqref{eq:generator_matrix}. To begin decoding, the receiver first recovers the six quadratic terms \( a_{34}, a_{24}, a_{14}, a_{23}, a_{13}, a_{12} \). Consider \( a_{12} \), which is the last message symbol (i.e., at position 11) in the message vector. We have the following:
\begin{align}
a_{12} &= \boldsymbol{a} \cdot \mathbf{e}_{11} = \boldsymbol{a}(c_1 + c_2 + c_3 + c_4) \notag \\
&= \boldsymbol{a}(c_5 + c_6 + c_7 + c_8) = \boldsymbol{a}(c_9 + c_{10} + c_{11} + c_{12}) \notag \\
&= \boldsymbol{a}(c_{13} + c_{14} + c_{15} + c_{16}) \label{eq:message_to_vector} \\
&= x_1 + x_2 + x_3 + x_4 = x_5 + x_6 + x_7 + x_8 \notag \\
&= x_9 + x_{10} + x_{11} + x_{12} = x_{13} + x_{14} + x_{15} + x_{16},
\label{eq:majority_vote}
\end{align}
where $c_j, 1 \le j \le 16$ are the columns of matrix $\boldsymbol{G}$. The four equalities in~\eqref{eq:majority_vote} yield four “votes” for \( a_{12} \). Equivalently,
\begin{align}
&\{x_1, x_2, x_3, x_4\},\quad\quad
\{x_5, x_6, x_7, x_8\},\notag\\
&\{x_9, x_{10}, x_{11}, x_{12}\},\quad
\{x_{13}, x_{14}, x_{15}, x_{16}\}
\label{eq:recovery_sets}
\end{align}
form four \emph{recovery sets} for the message symbol \( a_{12} \), where the sum of elements in each set equals \( a_{12} \). Among these sets, each codeword symbol \( x_i \) appears in exactly one vote. Formally, a recovery set for a message symbol is a set of codeword symbols that satisfies: (a) the sum of them equals the target message symbol, and (b) the set is \emph{minimal}, in the sense that no proper subset provides the same recovery. For example, while the set \( \{x_i : 1 \le i \le 12\} \) also yields \( a_{12} \) (since \( 3a_{12} = a_{12} \) in \( \mathbb{F}_2 \)), it is not minimal and can be inferred from the equalities in~\eqref{eq:majority_vote}; hence, it does not provide additional useful information for decoding.

Let $\boldsymbol{y} = \boldsymbol{x} + \boldsymbol{e}$ be the received vector, where $\boldsymbol{e} \in \mathbb{F}_2^{16}$ is a binary noise vector. If at most one error occurs, the recovery-set structure ensures that the error corrupts at most one of the four votes. For instance, if $y_1$ is flipped while all other symbols are correct, only the estimate $y_1 + y_2 + y_3 + y_4$ disagrees with the remaining three.
% \[
% y_5 + y_6 + y_7 + y_8,\quad
% y_9 + y_{10} + y_{11} + y_{12},\quad
% y_{13} + y_{14} + y_{15} + y_{16}.
% \]
A majority vote--that is, selecting the value that appears the most frequently, with $0$ favored in the event of a tie--correctly recovers the transmitted symbol \( a_{12} \). The same procedure applies to other quadratic terms \( a_{13}, a_{14}, a_{23}, a_{24}, a_{34} \), each associated with its own collection of recovery sets. This decoding method is known as \emph{majority-logic decoding} (MLD). It is important to note that decoding does not require listing all recovery sets; instead, one may select a subset of them to perform the majority vote. The particular choice of recovery sets used \emph{defines the decoder}.

To formalize the symbol ordering in this decoding algorithm, we define for each \( \ell \in \{0, 1, \dots, r\} \) a length-\( \ell \) tuple \( \sigma^{\ell} = \sigma_1 \sigma_2 \dots \sigma_{\ell} \), where the indices are drawn without replacement from the set \( \{1, 2, \dots, m\} \) and satisfy \( \sigma_1 < \sigma_2 < \dots < \sigma_{\ell} \). These tuples index the degree-\( \ell \) message symbols. For instance, the tuples \( 12, 13, 14, 23, 24, 34 \) index the second-degree tuples for \( m = 4 \). When \( \ell = 0 \), we define the special tuple \( \sigma^0 = (0) \), and denote \( a_{\sigma^0} := a_0 \). Hence, in our earlier example we have
\begin{itemize}
    \item Second-order symbols: \( a_{12}, a_{13}, a_{23}, a_{14}, a_{24}, a_{34} \).
    \item First-order symbols: \( a_1, a_2, a_3, a_4 \).
    \item Zero-order symbol: \( a_0 \).
\end{itemize}

Concretely, there are exactly \( \binom{m}{\ell} \) symbols of order \( \ell \). %We will see later, the order \( \ell \) of a symbol plays a central role in determining the structure and efficiency of its recovery sets.
Once the quadratic terms in~\eqref{eq:message_to_codeword} are recovered, their contributions are subtracted from the received codeword, reducing the problem to effectively decoding a first-order RM code. The decoder then proceeds iteratively to recover the remaining terms. This full sequential procedure is known as \emph{Reed's decoding}~\cite{reed1954class}. 

Reed's decoding method is an optimal decoder which minimizes the \textit{symbol error probability}, though not necessarily the word error probability. For this reason, although it is not a maximum-likelihood decoding algorithm, it aligns well with data recovery problems, where the goal is often to retrieve individual message symbols~\cite{SRR_RM_ISIT,SRR:preprint/arxiv/LySL25}.
% However, since all parity checks and votes for the second-degree terms can be computed in parallel, the first phase of this procedure is also known as \textbf{1-step majority-logic decoding}, in contrast to the full sequential Reed decoder.
\subsection{1-step Majority-logic Decoding}
When all message symbols of a codeword can be recovered simultaneously, decoding can be performed in a single step, unlike the sequential, multi-stage process used in Reed's algorithm. For example, consider the first-order RM code \( \mathrm{RM}(1, 2) \), whose generator matrix is
\[
\boldsymbol{G}_{\mathrm{RM}(1, 2)} =
\begin{bmatrix}
1 & 1 & 1 & 1\\
0 & 0 & 1 & 1\\
0 & 1 & 0 & 1
\end{bmatrix}
% = 
% \begin{bmatrix}
% \boldsymbol{1} \\
% \boldsymbol{v}_2 \\
% \boldsymbol{v}_1
% \end{bmatrix},
\]
which encodes a message vector \( \boldsymbol{a} := [a_0,\, a_2,\, a_1] \) into the codeword \( \boldsymbol{x} = \boldsymbol{a} \cdot \boldsymbol{G} := [x_1,\, x_2,\, x_3,\, x_4] \). The message symbols can each be recovered using two different votes
\begin{align}
a_0 &= \boldsymbol{a} \cdot \mathbf{e}_1 = \boldsymbol{a}\cdot c_1 = x_1 = x_2 + x_3 + x_4 \notag \\
    % &\quad\text{(since } \mathbf{e}_1 = c_1 = c_2 + c_3 + c_4) \notag \\
a_2 &= \boldsymbol{a} \cdot \mathbf{e}_2 = x_1 + x_3 = x_2 + x_4 \notag \\
% &\quad\text{(since } \mathbf{e}_2 = c_1 + c_3 = c_2 + c_4) \notag \\
a_1 &= \boldsymbol{a} \cdot \mathbf{e}_3 = x_1 + x_2 = x_3 + x_4. 
% &\quad\text{(since } \mathbf{e}_3 = c_1 + c_2 = c_3 + c_4) 
\label{eq:rm_first_order}
\end{align}

Therefore, upon receiving a possibly corrupted codeword \( \boldsymbol{y} \), one can apply majority voting using the paired expressions in~\eqref{eq:rm_first_order} to recover all three message symbols \( a_0, a_1, a_2 \) \emph{simultaneously}. This parallel recovery process is known as \emph{1-step majority-logic decoding} (1S-MLD)~\cite{Coding:books/PetersonW72}. Since errors in received codeword symbols affect all recovery sets (or votes) that contain them, the effectiveness of 1S-MLD critically depends on the overlap structure among these sets. Often, one seeks to have as many disjoint recovery sets as possible to ensure that each error corrupts at most one vote. For majority voting to succeed, the number of uncorrupted votes must exceed the number of corrupted ones under the given error pattern. Hence, the error-correction capability of 1S-MLD is fundamentally determined by the maximum number of such recovery sets and their constitution. The following result, which applies to all linear codes, establishes an upper bound on the number of errors that an 1S-MLD can correct.

% In Reed–Muller codes \( \mathrm{RM}(r, m) \), the first column \( c_1 \) of the generator matrix is always \( \mathbf{e}_1 \), so the codeword symbol \( x_1 \) alone serves as a recovery set of size one for the message symbol \( a_0 \). Using properties of the dual code, one can show that any other recovery set for \( a_0 \) must contain at least \( d^{\perp} - 1 \) codeword symbols, where \( d^{\perp} \) denotes the minimum distance of the dual code~\cite{SRR_Design:lySV2025}\hnote{This is our ``Combinatorial Designs" paper\hspace{4mm} $\uparrow$}. This constraint on recovery set size implies a limit on how many such sets can be constructed across the codeword symbols. As a result, we obtain the following bound, which applies to all linear \( (n, k) \) codes:

\begin{theorem}[see, e.g.,~{\cite[Thm.~10.2]{Coding:books/PetersonW72}}]
Let \( d^{\perp} \) denote the minimum distance of the dual code of an \( (n, k) \) linear code. Then the number of errors that can be corrected by 1S-MLD satisfies
\[
t \le \frac{n - 1}{2(d^{\perp} - 1)}.
\]
\label{thm:MJDecoding}
\end{theorem}
Theorem~\ref{thm:MJDecoding} yields the following consequence for Reed–Muller codes under 1S-MLD.
\begin{proposition}[Universal limitation for 1S-MLD on RM codes]
\label{prop:MLD_limit}
Let \( \mathrm{RM}(r,m) \) be a RM code with \(n=2^m\), \(d_{\min}=2^{m-r}\), and dual distance \(d^\perp = 2^{r+1}\).
Suppose a 1S-MLD is required to work uniformly for all parameters
\((r,m)\) with \(m\ge r+1\), and to correct any error pattern of weight at most \(t\). Then
\[
t \le \frac{d_{\min}}{4}.
\]
In particular, if a 1S-MLD guarantees correction of \(d_{\min}/4+\delta\) errors for some
\(\delta>0\), then it can only do so for a restricted set of parameter pairs \((r,m)\).
\end{proposition}

\begin{proof}
By Theorem~\ref{thm:MJDecoding}, any 1S-MLD for an \((n,k)\) linear code satisfies
\[
t \le \frac{n-1}{2(d^\perp-1)}.
\]
For \( \mathrm{RM}(r,m) \), this yields the exact bound
\begin{equation}\label{eq:exact-1s-bound}
t \le \frac{2^m-1}{2(2^{r+1}-1)}.
\end{equation}

Assume for contradiction that there exists a fixed \(\delta>0\) and a 1S-MLD that works for all \((r,m)\) (with \(m\ge r+1\)) and corrects
\begin{equation}\label{eq:assumption-better-than-quarter}
t \ge \frac{d_{\min}}{4}+\delta = 2^{m-r-2}+\delta.
\end{equation}
Combining \eqref{eq:assumption-better-than-quarter} with \eqref{eq:exact-1s-bound} gives, for all such \((r,m)\),
\[
\delta \le \frac{2^m-1}{2(2^{r+1}-1)} - 2^{m-r-2}
= \frac{2^{m-r-2}-\tfrac12}{2^{r+1}-1}.
\]
Now choose the valid family \(m=r+1\). Then \(d_{\min}=2\) and the right-hand side equals
\[
\frac{2^{(r+1)-r-2}-\tfrac12}{2^{r+1}-1}
= \frac{2^{-1}-\tfrac12}{2^{r+1}-1}
=0,
\]
forcing \(\delta\le 0\), a contradiction. Therefore, no 1S-MLD can guarantee correction beyond \(d_{\min}/4\) uniformly over all \((r,m)\).
\end{proof}

\begin{remark}
\textcolor{black}{Proposition~\ref{prop:MLD_limit} establishes an inherent limitation on the number of errors correctable by a 1S-MLD that works \emph{uniformly for all} Reed--Muller codes, irrespective
of the parameters \((r,m)\). Thus, any 1S-MLD that corrects more than \(d_{\min}/4\) errors cannot operate universally over all parameters \((r,m)\) with \(r\le m-1\). Conversely, the condition \(t \le d_{\min}/4\) is necessary but not sufficient for the existence of such a universal 1S-MLD.}

For specific RM codes—particularly those of low order—the general bound
\( \frac{n - 1}{2(d^{\perp} - 1)} \) may exceed \( d_{\min}/4 \), permitting decoding performance beyond the universal quarter-distance limit.
Accordingly, there exist 1S-MLD schemes that correct more than \( d_{\min}/4 \) errors for such codes; see, for example,~\cite{RMdecoding:journals/tit/RahmanI75}.
An more obvious, prominent instance is provided by Hadamard codes~\cite{eczoo_hadamard}, which coincide with first-order Reed--Muller codes.
\end{remark}
\section{Recovery Sets of All Message Symbols}\label{sec:recovery}
As demonstrated in the example with \( \mathrm{RM}(1, 2) \), designing an 1S-MLD for Reed–Muller codes involves identifying recovery sets for all message symbols simultaneously. This is different from Reed’s original multi-step decoding algorithm, where in each step, recovery sets are constructed only for message symbols of the current highest degree. 
Recent landmark results by the authors of~\cite{SRR_RM_ISIT,SRR:preprint/arxiv/LySL25} provide a complete characterization of such recovery sets of all message symbols for Reed--Muller codes. In this section, we briefly recall the key components of that result, beginning with the following theorem that underlies the Reed decoding algorithm; see, for example,{~\cite[Ch.~13, Thm.~14]{Coding:books/MacWilliamsS77}}.

% We now explore recovery sets of RM codes, their relationship, properties, and structure. These results lay the groundwork for deriving SRR properties in the next section. The following theorem formalizes recovery sets for message symbols of order $r$, as demonstrated by the majority-logic decoding described in previous section.
\iffalse
\begin{theorem} 
In \( \mathrm{RM}(r, m) \), each message symbol \( a_{\sigma^r} \) of order \( r \) can be recovered by partitioning the \( 2^m \) codeword coordinates (or symbols) of \( \boldsymbol{x} = \boldsymbol{a} \cdot \boldsymbol{G} \) into \( 2^{m - r} \) pairwise disjoint subsets, each of size \( 2^r \), such that the sum of the coordinates in each subset equals \( a_{\sigma^r} \).
\label{thm:ReedDecoding}
\end{theorem}
\fi

\begin{theorem}[Reed's decoding algorithm]\label{thm:ReedDecoding} In \( \mathrm{RM}(r, m) \), each message symbol of order $r$ (the highest order) \( a_{\sigma^r} \) can be determined by partitioning the \( 2^m \) coordinates of the codeword \( \boldsymbol{x} = \boldsymbol{a} \cdot \mathbf{G} \) into \( 2^{m-r} \) pairwise disjoint subsets of size \( 2^r \), where the sum of the coordinates within each subset equals \( a_{\sigma^r} \). Each such set is a recovery set for $a_{\sigma^r}$.
\end{theorem}

This theorem generalizes the equalities shown in~\eqref{eq:majority_vote}. Geometrically, the coordinates constituting each recovery set correspond to the points of an \( r \)-flat in \( \mathrm{EG}(m, 2) \). Specifically, if we denote the recovery sets as $S_1, S_2, \dots, S_{2^{m-r}}$, then for each $i$, the point set $\{P_j \mid \, j \in S_i\}$ forms an $r$-flat. The flat containing the origin \( P_1 \) is an \( r \)-dimensional linear subspace, while the remaining \( 2^{m - r} - 1 \) flats are its \emph{affine translations}. For example, in \( \mathrm{RM}(2, 4) \), the four recovery sets in~\eqref{eq:recovery_sets} correspond to the $2$-flats  \( \{P_1, P_2, P_3, P_4\} \), \( \{P_5, P_6, P_7, P_8\} \), \( \{P_9, P_{10}, P_{11}, P_{12}\} \), and \( \{P_{13}, P_{14}, P_{15}, P_{16}\} \). The first set, which contains the origin $P_1$, is a $2$-dimensional linear subspace, and the others are its affine translations. Note that they are pairwise disjoint.
% We note here that the sequential nature of the Reed algorithm, despite its implementation advantages, poses a significant issue for data retrieval: recovering lower-degree message symbols requires first decoding higher-degree symbols and subtracting their contributions. This dependency creates inefficiencies in storage systems with dynamic demand patterns. Before this work, direct recovery methods that bypass this sequential constraint were unknown, despite being essential in practice.

The following three theorems provide a deeper understanding of the recovery sets for message symbols of arbitrary order \( \ell \in \{0, 1, \dots, r\} \), and in fact generalize Theorem~\ref{thm:ReedDecoding}. Informally, the authors in~\cite{SRR_RM_ISIT,SRR:preprint/arxiv/LySL25} showed that for each message symbol \( a_{\sigma^{\ell}} \) of order \( \ell \), the following properties hold:

% \begin{itemize}
%     \item There exists exactly one recovery set of size \( 2^{\ell} \). The points corresponding to the \( 2^{\ell} \) coordinates in this set form a particular \( \ell \)-dimensional subspace \( \mathscr{S} \subset \mathrm{EG}(m, 2) \) whose construction is known explicitly (Theorem~\ref{thm:RecoverySet}).
    
%     \item All other recovery sets for \( a_{\sigma^{\ell}} \) have size at least \( 2^{r+1} - 2^{\ell} \) (Theorem~\ref{thm:RecoverySetSize}). Among them, exactly \( \gaussbinom{m - \ell}{r + 1 - \ell} \) have size precisely \( 2^{r+1} - 2^{\ell} \). Each such set arises as the complement of \( \mathscr{S} \) within a distinct \( (r+1) \)-dimensional subspace of \( \mathrm{EG}(m, 2) \) that contains \( \mathscr{S} \), establishing a one-to-one correspondence between these subspace complements and the recovery sets. Moreover, each point \( P_j \in \mathrm{EG}(m, 2) \setminus \mathscr{S} \) lies in exactly \( \gaussbinom{m - \ell - 1}{r - \ell} \) such complements (Theorem~\ref{thm:RecoverySetCount}).
% \end{itemize}
\begin{itemize}
    \item There exists a unique recovery set $S$ of size \( 2^{\ell} \). The points $\{P_j \mid j \in S\}$ form a specific \( \ell \)-dimensional linear subspace \( \mathscr{S} \subset \mathrm{EG}(m, 2) \) whose construction is explicitly known (Theorem~\ref{thm:RecoverySet}).
    
    \item All other recovery sets for \( a_{\sigma^{\ell}} \) have size at least \( 2^{r+1} - 2^{\ell} \) (Theorem~\ref{thm:RecoverySetSize}). Among these, exactly \( \gaussbinom{m - \ell}{r + 1 - \ell} \) attain the minimum size of \( 2^{r+1} - 2^{\ell} \). For each such recovery set \( T \), the points \( \{ P_j \mid j \in T \} \) forms a complement of \( \mathscr{S} \) within a distinct \( (r+1) \)-dimensional linear subspace containing \( \mathscr{S} \). This establishes a one-to-one correspondence between these subspace complements and the recovery sets. Furthermore, these sets form a \emph{combinatorial $t$-design} with $t=1$: each point \( P_j \in \mathrm{EG}(m, 2) \setminus \mathscr{S} \) is included in exactly \( \gaussbinom{m - \ell - 1}{r - \ell} \) such complements (Theorem~\ref{thm:RecoverySetCount}). (Recall that a \( t\text{--}(n,k,\lambda) \) block design (or a \( t \)-design) is a combinatorial structure consisting of a set \( V \) of \( n \) elements (called \emph{points}) and a collection of $k$-element subsets of \( V \) (called \emph{blocks}), such that every \( t \)-subset of \( V \) is contained exactly in \( \lambda \) blocks.)
\end{itemize}

% We now characterize the structure of recovery sets for message symbols associated with arbitrary degrees, providing a mathematical framework that forms the recovery sets for message symbols. %without the sequential constraints of traditional decoding.

\begin{theorem}[\hspace{-0.1mm}{\cite[Thm.~4]{SRR_RM_ISIT,SRR:preprint/arxiv/LySL25}}]
For any integer \( \ell \) with \( 1 \leq \ell \leq r \), the symbol \( a_{\sigma^{\ell}} \) can be recovered by summing the codeword coordinates indexed by a specific coordinate subset \( S \subseteq [2^m] \). This set satisfies:

\vspace{2mm}
% \begin{align}\label{min_recovery_set}
\(
 \quad   S \ni 1, \quad |S| = 2^{\ell},\,\, \text{ and }\,\, \sum\limits_{j \in S} x_j = a_{\sigma^{\ell}}. 
    \vspace{2mm}
\)
% \end{align} 

We refer to \( S \) as a \textit{recovery set} for \( a_{\sigma^{\ell}} \). Geometrically, the points \( \{ P_j \mid j \in S \} \) form an \( \ell \)-dimensional linear subspace \( \mathscr{S} \subset \mathrm{EG}(m, 2) \) (or, an $\ell$-flat passing through the origin $P_1$).% We refer to \( S \) as a \textit{recovery set} for the symbol \( a_{\sigma^{\ell}} \).
\label{thm:RecoverySet}
\end{theorem}

% Having established the existence of recovery sets for each symbol \( a_{\sigma^l} \), 
% The following theorem determines their sizes, showing that smaller recovery sets are uniquely determined, while larger sets are subject to specific size constraints.
\begin{theorem}[\hspace{-0.1mm}{\cite[Thm.~5]{SRR_RM_ISIT,SRR:preprint/arxiv/LySL25}}]\label{thm:RecoverySetSize}
    If \( {\ell} < r \), then \( S \) is the \textbf{unique} recovery set for \( a_{\sigma^{\ell}} \) with size less than \( 2^r \). Any other recovery set must have a size of at least \( 2^{r+1} - |S| = 2^{r+1} - 2^{\ell} \). 
    
    In the case \( {\ell} = r \), the set \( S \) has size \( 2^r \), which matches the size of all other recovery sets for \( a_{\sigma^r} \).
\end{theorem}

% With the existence of recovery sets established, the following theorem quantifies their count and distribution, using the Gaussian binomial coefficient to determine the frequency of inclusion across different coordinates.
\begin{theorem}[\hspace{-0.1mm}{\cite[Thm.~6]{SRR_RM_ISIT,SRR:preprint/arxiv/LySL25}}]\label{thm:RecoverySetCount}
     For each symbol \( a_{\sigma^{\ell}} \), there are exactly \( \gaussbinom{m - {\ell}}{r + 1 - {\ell}} \) recovery sets of size \( 2^{r+1} - 2^{\ell} \). Each such set corresponds to the complement of \( \mathscr{S} \) within a distinct \( (r+1) \)-dimensional linear subspace of \( \mathrm{EG}(m, 2) \) containing \( \mathscr{S} \) (and is therefore disjoint from $S$).
    
     Furthermore, these recovery sets are distributed evenly across the remaining coordinates: every coordinate \( x_j \) with \( j \notin S \) appears exactly in \( \gaussbinom{m - {\ell} - 1}{r - {\ell}} \) of these sets, forming a 1-design (i.e., a $t$-design with $t=1$) over the set $[n]\setminus S$.
\end{theorem}
\begin{remark}
    The last three theorems generalize Theorem~\ref{thm:ReedDecoding}, which appears as the special case \(\ell = r\). In this case,  
\(
\gaussbinom{m-\ell}{\,r+1-\ell\,} = \gaussbinom{m-r}{1} = 2^{m-r} - 1,
\)
which equals the number of flats obtained as affine translations of the original \( r \)-dimensional subspace. Moreover, each coordinate lies in exactly  
\(
\gaussbinom{m-r-1}{0} = 1
\)
of these recovery sets, meaning the sets are pairwise disjoint.
\end{remark}
% Theorem~\ref{thm:RecoverySetCount} not only quantifies the number of recovery sets of size $2^{r+1}-2^l$ for each message symbol but also elucidates the relationship between these recovery sets and the minimum-weight codewords in the dual Reed--Muller code. 

We conclude this section with an example that illustrates the identification and enumeration of recovery sets in $\mathrm{RM}(2, \,4)$. 
\begin{exmp}\label{ex:RM24}
    Consider the Reed--Muller code \( \mathrm{RM}(2, 4) \). This example aims to identify recovery sets for symbol \( a_1 \). First, observe that the symbol \( a_1 \) is given by
    \[
    a_1 = \boldsymbol{a} \cdot \mathbf{e}_5 = \boldsymbol{a}(c_1 +c_2) = x_1 + x_2,
    \]
    indicating that \( S = \{1, 2\} \) is a recovery set for \( a_1 \) with size 2. Consequently, $\mathscr{S} = \{P_1, P_2\}$ is a $1$-dimensional subspace in $\mathrm{EG}(m, 2)$.
    % Observe that symbol \( a_1 \) can be given by:
    % \(
    % a_1 = \boldsymbol{a} \cdot \mathbf{e}_1 = x_1,
    % \)
    % indicating that \( S = \{1\} \) is a recovery set for \( a_0 \) of size 1.     
    Additionally, \( a_1 \) admits multiple representations as a linear combination of other coordinates (constructions of these sets are given in the proofs of Theorems~4-6 in~\cite{SRR:preprint/arxiv/LySL25})
    \begin{alignat*}{2}
        a_1 & = x_1 + x_2, && \quad (S)   \\
        & = x_3 + x_4 + x_5 + x_6 + x_7 + x_8, && \quad (S_1)   \\
        & = x_3 + x_4 + x_9 + x_{10} + x_{11} + x_{12}, && \quad (S_2)   \\
            &= x_5 + x_6 + x_9 + x_{10} + x_{13} + x_{14}, && \quad (S_3)   \\
            & = x_5 + x_6 + x_{11} + x_{12} + x_{15} + x_{16}, && \quad (S_4)   \\
            &= x_3 + x_4 + x_{13} + x_{14} + x_{15} + x_{16}, && \quad (S_5)   \\
            & = x_7 + x_8 + x_{9} + x_{10} + x_{15} + x_{16}, && \quad (S_6)   \\
            &= x_7 + x_8 + x_{11} + x_{12} + x_{13} + x_{14}. && \quad (S_7)   
    \end{alignat*}
    Concretely, there are $\gaussbinom{4-1}{2+1-1} = \gaussbinom{3}{2} = 7$ recovery sets of size 6 for $a_1$. Moreover, each of the 14 coordinates $x_j,\, j \in \{3, 4, \dots, 16\}$ appears in exactly $\gaussbinom{2}{1} = 3$ of them, satisfying the identity \( 7 \cdot 6 = 14 \cdot 3 \).
    
    Figure~\ref{fig:flower} illustrates this geometric structure as a "flower". The subspace \( \mathscr{S} \) is represented as the red kernel. Each blue petal (e.g., \( S_1, S_2 \)) represents a distinct recovery set of size 6, which forms the complement of \( \mathscr{S} \) within a 3-dimensional subspace containing it. Also, \( S_1 \cap S_2 = \{x_3, x_4\}\).
    
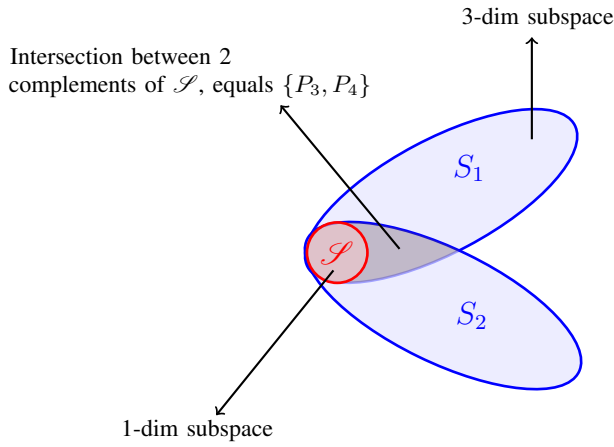
\begin{figure}[ht]
        \centering
        \begin{tikzpicture}[scale=0.5]
%--- first petal  S1  (3-dim subspace) ----------------------------------------
\begin{scope}[shift={(3.7,2)},rotate=28]                       % centre+tilt
  \filldraw[fill=blue!10,fill opacity=0.7,draw=blue,line width=1] (0,0)
            ellipse (4.0 and 1.5);
  \node[blue,font=\large\bfseries] at (1.0,0.3) {$S_{1}$};
\end{scope}

--- second petal  S2  ---------------------------------------------------------
\begin{scope}[shift={(3.75,-0.9)},rotate=-26]
  \filldraw[fill=blue!10,fill opacity=0.7,draw=blue,line width=1] (0,0)
            ellipse (4.0 and 1.5);
  \node[blue,font=\large\bfseries] at (0.8,0.1) {$S_{2}$};
\end{scope}

\begin{scope}
  % 1. keep only the part that lies in the first ellipse
  \clip[shift={(3.7,2)},rotate=28] (0,0) ellipse (4.0 and 1.5);
  
  % 2. inside that clipped region, draw the second ellipse
  %    -> the fill appears exactly on the intersection
  \fill[teal!60!red, fill opacity=0.3]      % <- choose any colour you like
        [shift={(3.75,-0.9)},rotate=-26]
        (0,0) ellipse (4.0 and 1.5);
\end{scope}

%--- kernel (1-dim subspace) --------------------------------------------------
% \filldraw[fill=red!20,draw=red,line width=1] (0,0) circle (1);   % red disk
\filldraw[fill=red!20, fill opacity=0.5, draw=red, line width=1] (0.95,0.5) circle (0.8);
\node[red,font=\large\bfseries] at (0.95,0.5) {$\mathscr{S}$};

%--- annotation arrows ---------------------------------------------------------
\draw[thick,->] (0.85,0.03) -- (-2.25,-3.8);%.. controls (-1.4,-1.0) and (-0.6,-0.4) .. (-0.2,-0.1);
\node[anchor=east] at (-0.5,-4.2) {\small 1-dim subspace};

\draw[thick,->] (6.1,3.5) -- (6.1,6.2);
\node[anchor=west] at (4.0,6.7) {\small 3-dim subspace};

\draw[thick,->] (2.6,0.6) -- (-0.55,4.4);%.. controls (-1.4,-1.0) and (-0.6,-0.4) .. (-0.2,-0.1);
\node[anchor=east] at (-1.6,5.7) {\small Intersection between 2};
\node[anchor=east] at (2.1,4.9) {\small complements of $\mathscr{S}$, equals $\{P_3, P_4\}$};

%--- a few scatter points (other coordinates) ----------------------------------
% \foreach \x/\y in {-0.6/3.2, 0.3/2.8, 1.4/3.4, 2.4/2.7, 4.4/0.4,
%                    -1.8/-0.2, -2.2/0.3, 1.0/-3.4, -0.6/-3.2}
%   \fill (\x,\y) circle (1.4pt);

\end{tikzpicture}

        \caption{A schematic “flower” structure: the kernel (red) is the 1-dimensional subspace \( \mathscr{S} \), while each petal (blue) represents a recovery set \( S_j \) formed as the complement of \( \mathscr{S} \) in a 3-dimensional subspace.}
        \label{fig:flower}
    \end{figure}
\end{exmp}

\section{One-step Majority-logic Decoding of RM codes}\label{sec:RM_MLD}
We next describe how we use the recovery sets described in the previous section to construct an 1S-MLD procedure for $\text{RM}(r,m)$. Specifically, for each message symbol \( a_{\sigma^{\ell}} \) of order \( \ell \in \{0, 1, \dots, r\} \), we use the following recovery sets in the decoding process.
% \begin{tcolorbox}[colback=black!5!white,colframe=black!75!black,title=Recovery sets used for 1S-MLD]
\begin{myrefbox}[box:recovery]{Recovery sets used for 1S-MLD}
  \begin{itemize}[leftmargin=*]
    \item A unique recovery set $S$ of size \( 2^{\ell} \), whose corresponding points \( \{ P_j \mid j \in S \} \) form an \( \ell \)-dimensional linear subspace \( \mathscr{S} \subset \mathrm{EG}(m, 2) \).
    \item Exactly \( \gaussbinom{m - \ell}{r + 1 - \ell} \) recovery sets of size precisely \( 2^{r+1} - 2^{\ell} \). For each such recovery set \( T \), the points \( \{ P_j \mid j \in T \} \) forms a complement of \( \mathscr{S} \) within a distinct \( (r+1) \)-dimensional linear subspace containing \( \mathscr{S} \). Moreover, each codeword symbol \( x_j, \, j \notin S \), belongs to exactly \( \gaussbinom{m - \ell - 1}{r - \ell} \) of these recovery sets.
\end{itemize}
\end{myrefbox}
We now investigate the maximum number of errors and erasures that this 1S-MLD can correct, starting with the following example.
\begin{exmp}
Consider \( \mathrm{RM}(2, 4) \) in Example~\ref{ex:RM24} whose minimum distance is \( d_{\min} = 2^{4 - 2} = 4 \). The symbol \( a_1 \) has eight recovery sets: one of size 2 (denoted \( S \)) and seven of size 6 (denoted \( S_1 \) to \( S_7 \)). Each coordinate \( x_j \), for \( j \in \{3, 4, \dots, 16\} \), appears in exactly three of the size-6 recovery sets, while \( x_1 \) and \( x_2 \) appear only in the size-2 set. As a result, the 1S-MLD can correct any single error. However, if two errors occur (e.g., at \( x_3 \) and \( x_5 \)), they may corrupt five out of the eight votes, causing  MLD to fail.

    In contrast, to make all recovery sets unusable by erasures, at least four coordinates must be erased (e.g., \( x_1, x_3, x_5, x_7 \)), ensuring that each recovery set contains at least one erased symbol. An exhaustive search (e.g., by a simple computer program) confirms that any three or fewer erasures will always leave at least one recovery set unerased. %This minimum number of erasures can be confirmed by a simple computer search.

    This example demonstrates that the 1S-MLD achieves its theoretical limits for \( \mathrm{RM}(2, 4) \): it can correct up to \( d_{\min}/4 = 1 \) error and up to \( d_{\min} - 1 = 3 \) erasures when recovering \( a_1 \).
\end{exmp}

As suggested in the example above, the following theorem shows that this decoding method achieves optimal error-correction performance under the 1S-MLD setting: specifically, it can correct any error pattern of weight up to \( d_{\min} / 4 = 2^{m - r - 2} \). 
\begin{theorem} [Optimal Error Correction]
    For Reed--Muller code \( \mathrm{RM}(r, m) \), 1S-MLD using the recovery sets described in Box~\ref{box:recovery} can correct up to \( d_{\min}/4 = 2^{m - r - 2} \) errors.
    \label{thm:optimal_error_decoder}
\end{theorem}
\begin{proof}
    For each message symbol of an arbitrary order $\ell \in [0,\, r]$, the total number of recovery sets (i.e., votes) for \( a_{\sigma^{\ell}} \) is 
    \[
        1 + \gaussbinom{m - \ell}{r + 1 - \ell}.
    \]
    We show that when at most \( 2^{m - r - 2} \) errors occur, fewer than half of these votes are corrupted, ensuring correct decoding via majority voting. Indeed, each error at a codeword coordinate \( j \) corrupts exactly one vote if \( j \in S \), or exactly \( \gaussbinom{m - \ell - 1}{r - \ell} \) votes if \( j \notin {S} \). Therefore, the worst-case total number of corrupted votes caused by \( 2^{m - r - 2} \) errors is
    \[
        2^{m - r - 2} \cdot \gaussbinom{m - \ell - 1}{r - \ell}.
    \]

    We now compare this quantity to half the number of available votes. Observe that
    \begin{align*}
        \frac{\gaussbinom{m - \ell}{r + 1 - \ell}}{\gaussbinom{m - \ell - 1}{r - \ell}} 
        &= \frac{\prod\limits_{i=0}^{r - \ell} (1 - 2^{m - \ell - i}) / (1 - 2^{i + 1})}
                {\prod\limits_{i=0}^{r - \ell - 1} (1 - 2^{m - \ell - 1 - i}) / (1 - 2^{i + 1})} \\
        &= \frac{2^m - 2^{\ell}}{2^{r+1} - 2^{\ell}} \\
        &\ge \frac{2^m - 1}{2^{r+1} - 1}, \quad \text{for all } \ell \in \{0, 1, \dots, r\} \\
        &\ge 2^{m - r - 1}, \quad \text{for all } m \ge r + 1.
    \end{align*}
    It follows that
    \begin{align*}
        2^{m - r - 2} \cdot \gaussbinom{m - \ell - 1}{r - \ell}
         & \le \frac{1}{2} \gaussbinom{m - \ell}{r + 1 - \ell} \\
         & < \frac{1}{2} \left(\gaussbinom{m - \ell}{r + 1 - \ell}+1\right).
    \end{align*}
    Hence, strictly less than half of the votes are corrupted, and the majority-logic decoder correctly recovers the message symbol \( a_{\sigma^{\ell}} \) for each \( \ell \in \{0, 1, \dots, r\} \). 
\end{proof}

The next question is whether the proposed decoder is also optimal in correcting erasures, that is, whether it can recover from any erasure pattern of weight up to \( d_{\min} - 1 \). Using the same reasoning as in the proof of Theorem~\ref{thm:optimal_error_decoder}, one can initially argue that at least \( 2^{m - r - 1} = d_{\min} / 2 \) erasures are required to disable all votes, which implies that the decoder can correct up to \( 2^{m - r - 1} -1\) erasures. However, this estimate is conservative. 

We demonstrate below that using the geometric structure of the recovery sets enables us to significantly enhance this erasure correction capability. We first, in Theorem~\ref{thm:optimal_transversal}, establish a result on the minimum size of a transversal for truncated $r$-flats pertaining to the "flower" structure discussed above.  We then use this result and the geometric properties of RM codes to prove that our 1S-MLD (based on the recovery sets in Box~\ref{box:recovery}) can correct up to exactly $d_{\min}-1$ erasures at Theorem~\ref{thm:optimal_erasure_decoder}.

\begin{theorem}[Optimal Transversal/Blocking set for Truncated $r$-Flats]
\label{thm:optimal_transversal}
Let \( V = \mathrm{EG}(m, 2) = \mathbb{F}_2^m \), and fix a subspace \( \mathscr{S} \subseteq V \) with \( \dim \mathscr{S} = \ell < m \). For any integer \( r \) satisfying \( \ell < r \le m \), define the collection
\[
\mathcal{C} := \left\{ F \setminus \mathscr{S} \mid F \le V,\, \dim F = r,\, \mathscr{S} \subseteq F \right\}.
\]
That is, \( \mathcal{C} \) consists of the complements of \( \mathscr{S} \) within all
\(r\)-subspaces of \(V\) that contain \( \mathscr{S} \). Then the minimum size of a set that intersects every member of \( \mathcal{C} \) (i.e., the transversal number) is
\[
\tau(\mathcal{C}) = 2^{m - r + 1} - 1,
\]
and one natural example of such a minimum transversal is the set of nonzero vectors of a subspace \( U \le V \) with \( \dim U = m - r + 1 \) and \( U \cap \mathscr{S} = \{0\} \).
\end{theorem}

\begin{proof}
We proceed in four steps. First, we reduce the problem to a quotient space \( W = V / \mathscr{S} \), translating the task of intersecting affine subspace complements into a condition on subspaces of \( W \) (Step 1). We then reinterpret this condition in the setting of projective geometry \( \mathrm{PG}(n, 2) \), converting the transversal problem into a classical blocking set problem (Step 2). In Step 3, we apply a result of Bose and Burton to identify the minimal transversal size in \( \mathrm{PG}(n, 2) \). Finally, we lift the result back to \( V \), showing that the preimage of such a transversal corresponds to the claimed optimal transversal (Step 4).
\paragraph*{\textbf{Step 1 (Quotient Reduction)}}
 Let \( W := V / \mathscr{S} \cong \mathbb{F}_2^{m - \ell} \), and denote by \( \pi : V \to W \) the canonical projection \( \pi(v) = v + \mathscr{S} \). For each subspace \( F \supseteq \mathscr{S} \) with \( \dim F = r \), the image \( \bar{F} := F / \mathscr{S} \) is an \( (r - \ell) \)-dimensional subspace of \( W \), and we have
\[
\pi(F \setminus \mathscr{S}) = \bar{F} \setminus \{0\}.
\]
A set \( T \subseteq V \) intersects every set in \( \mathcal{C} \) if and only if its image under the projection intersects every \( (r - \ell) \)-dimensional subspace of \( W \). Specifically, the set \( B := \pi(T \setminus \mathscr{S}) \subseteq W \setminus \{0\} \) must intersect every such subspace.

\paragraph*{\textbf{Step 2 (Projectivization)}} We identify the nonzero vectors of \( W \) with the points of the projective space \( \mathrm{PG}(n, 2) \), where \( n := m - \ell - 1 \). The subspaces \( \bar{F} \) of dimension \( r - \ell \) in \( W \) correspond to \( k \)-dimensional subspaces in \( \mathrm{PG}(n, 2) \), with \( k := r - \ell - 1 \). Thus, the problem reduces to the classical transversal identification problem of finding the minimal number of points in \( \mathrm{PG}(n, 2) \) that intersect every \( k \)-dimensional subspace.

\paragraph*{\textbf{Step 3 (Bose–Burton Theorem Application)}} According to a result by Bose and Burton~\cite{BOSE196696} (see also~\cite{Blocking_subspaces}), the smallest transversal with respect to \( k \)-dimensional subspaces in \( \mathrm{PG}(n, 2) \) is the set of points contained in an \( (n - k) \)-dimensional subspace. In our case, this yields a minimal transversal of size
\[
\theta_{n - k}(2) = \frac{2^{n - k + 1} - 1}{2 - 1} = 2^{m - r + 1} - 1 = \theta_{m-r}(2).
\]

\paragraph*{\textbf{Step 4 (Lifting to \( V \))}} Let $U' \subseteq W$ be an $(m-r+1)$-dimensional subspace that defines a minimum transversal in $\mathrm{PG}(n,2)$, i.e., $U' \setminus \{0\}$. 
Choose a linear complement $U \subseteq V$ of $\mathscr{S}$ such that $\pi(U)=U'$; then $\dim U = m-r+1$ and $U \cap \mathscr{S} = \{0\}$. 
The set 
\[
T := U \setminus \{0\}
\]
projects bijectively onto $U' \setminus \{0\}$, and hence $T$ intersects every member of $\mathcal{C}$. 
This proves that a transversal of size $2^{\,m-r+1}-1$ exists in $V$.  

More generally, any transversal in \( V \) must project to a transversal in \( \mathrm{PG}(n,2) \), so its size is at least \( 2^{\,m-r+1}-1 \). 
Conversely, once \( U' \) is fixed, one can obtain a minimum transversal in \( V \) by selecting \emph{exactly one representative from each nonzero coset} in \( \pi^{-1}(U') \). 
Such a choice need not be a linear subspace of \( V \), but it still yields a transversal of size \( 2^{\,m-r+1}-1 \). 

Thus, every minimum transversal in \( V \) corresponds to a set whose projection equals \( U' \setminus \{0\} \) for some \( (m-r+1) \)-subspace \( U' \le W \). 
In particular, one always exists of the form \( U \setminus \{0\} \) with \( U \cap \mathscr{S}=\{0\} \).
\end{proof}

\begin{remark}
\leavevmode
\begin{enumerate}[label=\textup{(\roman*)}]
    \item \textcolor{black}{The transversal number $\tau(\mathcal{C})$ is independent of the dimension $\ell$ of the kernel subspace.}
  \item When \( \ell = 0 \), Theorem~\ref{thm:optimal_transversal} recovers the classical result for blocking all \( r \)-flats in \( \mathrm{EG}(m, 2) \).
  % \item When \( r = m - 1 \), the result implies \( \tau = 2 \), achieved by the two nonzero points on any line through the origin that does not intersect \( \mathscr{S} \).
  \item The argument extends to general finite fields \( \mathbb{F}_q \), replacing \( 2^{m - r + 1} - 1 \) with \( \theta_{m - r}(q) = \frac{q^{m - r + 1} - 1}{q - 1} \).
\end{enumerate}
\end{remark}

\begin{theorem}[Optimum Erasure Decoding]
    For \( \mathrm{RM}(r, m) \), the 1S-MLD using the recovery sets described in Box~\ref{box:recovery} can correct any erasure pattern of weight up to \( d_{\min} - 1 = 2^{m - r} - 1 \).
    \label{thm:optimal_erasure_decoder}
\end{theorem}
\begin{proof}
Recall that there are \( \gaussbinom{m - \ell}{\,r + 1 - \ell\,} \) recovery sets of size \( 2^{r+1} - 2^{\ell} \), each arising as the complement of \( \mathscr{S} \) within a distinct \( (r+1) \)-dimensional subspace of \( \mathrm{EG}(m, 2) \) containing \( \mathscr{S} \). 
By Theorem~\ref{thm:optimal_transversal}, the smallest set of coordinates that blocks all such recovery sets—while leaving the \( 2^\ell \)-sized recovery set unblocked—has cardinality \( 2^{m-r} - 1 \). Equivalently, any set of at most \( 2^{m-r} - 1 \) erased coordinates leaves at least one recovery set unblocked. This concludes the proof.
\end{proof}
\textcolor{black}{
From Theorems~\ref{thm:optimal_error_decoder} and~\ref{thm:optimal_erasure_decoder}, we observe a distinct behavior between error and erasure correction. For a message symbol of order $\ell$ ($1\le \ell \le r$), the number of erasures required to block all recovery sets is constant at $2^{m-r}$, independent of order $\ell$ of the message symbol. In contrast, the maximum number of errors that can be corrected by majority-logic is given by
\[
\frac{\left[ \begin{smallmatrix} m - \ell \\ r + 1 - \ell \end{smallmatrix} \right]_2}{2\left[ \begin{smallmatrix} m - \ell - 1 \\ r - \ell \end{smallmatrix} \right]_2}  
= \frac{2^m-2^{\ell}}{2(2^{r+1}-2^{\ell})},
\]
which is a strictly increasing function of $\ell$ for fixed parameters $r$ and $m$. This observation leads to the following corollary.
\begin{corollary}
For a message symbol of order $\ell$, the correction capability against erasures (which is $2^{m-r}$) is constant across all values of $\ell$, whereas the correction capability against errors $\left(\text{which is }\frac{\gaussbinom{m - \ell}{r + 1 - \ell}}{2\gaussbinom{m - \ell - 1}{r - \ell}}\right)$ increases with $\ell$.
\end{corollary}
}
% \end{proof}
% \subsection{Complexity analysis}
\section{Conclusion}\label{sec:conclusion}
This work introduces the first one-step majority-logic decoding method for Reed--Muller codes that applies to all code parameters \( r \) and \( m \). We prove that the proposed decoder achieves optimal performance in both error and erasure correction within the one-step decoding framework. A natural direction for future work is to investigate the use of 1S-MLD for other classes of codes, including \( q \)-ary Reed--Muller codes, and Coxeter codes.
\section*{Acknowledgment}
The authors thank V.~Lalitha, Michael Schleppy, and Andrea Di Giusto for valuable discussions and observations.
\balance
\bibliography{bibliography}
\bibliographystyle{IEEEtran}

% \begin{IEEEbiography}{Michael Shell}
% Biography text here.
% \end{IEEEbiography}

\end{document}